\documentclass[letterpaper, 10pt, conference]{ieeeconf}

\IEEEoverridecommandlockouts
\renewcommand{\baselinestretch}{.975}
\usepackage{amsmath, amssymb, amsfonts, mathtools} 
\usepackage{graphicx}
\usepackage{xcolor}
\usepackage{hyperref}
\hypersetup{colorlinks}
\usepackage[capitalize]{cleveref}
\usepackage{cite}
\usepackage{float}
\usepackage{subcaption}
\usepackage[normalem]{ulem}
\usepackage{xcolor}
\definecolor{darkgreen}{RGB}{0,130,0}
\definecolor{myorange}{HTML}{FFAB6B}
\usepackage[table]{xcolor}
\usepackage{booktabs}
\usepackage{algorithm}
\usepackage{algorithmic}
\usepackage{amsmath,amssymb}
\usepackage{bm}
\usepackage{soul}

\newtheorem{prop}{Proposition}
\newtheorem{lem}{Lemma}

\newtheorem{assume}{Assumption}
\newtheorem{rem}{Remark}

\usepackage{subcaption}
\def \mb{\mathbf}
\def \bb{\mathbb}

\def \ms{\mathsf}
\def \mf{\mathfrak}
\def \mc{\mathcal}

\def \ol{\overline}
\def \wt{\widetilde}

\newcommand\eqdef\coloneqq
\usepackage{stfloats}

\title{\LARGE \bf
    HyperKKL: Learning KKL Observers for Non-Autonomous Nonlinear Systems via Hypernetwork-Based Input Conditioning
    }

\author{Yahia Salaheldin Shaaban$^{1}$, Abdelrahman Sayed Sayed$^{2}$, M. Umar B. Niazi$^{3}$, and Karl Henrik Johansson$^{3}$%
\thanks{$^{1}$Mohamed bin Zayed University of Artificial Intelligence (MBZUAI), Abu Dhabi, UAE.
        Email: {\tt\small yahia.shaaban@mbzuai.ac.ae}}%
\thanks{$^{2}$Univ Gustave Eiffel, COSYS-ESTAS, F-59657 Villeneuve d'Ascq, France.
        Email: {\tt\small abdelrahman.ibrahim@univ-eiffel.fr}}%
\thanks{$^{3}$Department of Decision and Control Systems and Digital Futures, KTH Royal Institute of Technology, SE-100 44 Stockholm, Sweden.
        Email: {\tt\small \{mubniazi, kallej\}@kth.se}}%
}

\begin{document}

\maketitle
\thispagestyle{empty}
\pagestyle{empty}

\begin{abstract}
Kazantzis–Kravaris/Luenberger (KKL) observers are a class of state observers for nonlinear systems that rely on an injective map to transform the nonlinear dynamics into a stable quasi-linear latent space, from where the state estimate is obtained in the original coordinates via a left inverse of the transformation map. Current learning-based methods for these maps are designed exclusively for autonomous systems and do not generalize well to controlled or non-autonomous systems. 
In this paper, we propose two learning-based designs of neural KKL observers for non-autonomous systems whose dynamics are influenced by exogenous inputs. To this end, a hypernetwork-based framework (HyperKKL) is proposed with two input-conditioning strategies. First, an augmented observer approach (HyperKKL\textsubscript{obs}) adds input-dependent corrections to the latent observer dynamics while retaining static transformation maps. Second, a dynamic observer approach (HyperKKL\textsubscript{dyn}) employs a hypernetwork to generate encoder and decoder weights that are input-dependent, yielding time-varying transformation maps.
We derive a theoretical worst-case bound on the state estimation error. 
Numerical evaluations on four nonlinear benchmark systems show that input conditioning yields consistent improvements in estimation accuracy over static autonomous maps, with an average symmetric mean absolute percentage error (SMAPE) reduction of 29\% across all non-zero input regimes. 
\end{abstract}


\section{Introduction}

Reconstructing the full state of a dynamical system from partial measurements is a foundational problem in control theory. In most applications, only a subset of the state variables is directly measurable, and the remaining states must be estimated by an observer. However, designing observers for nonlinear systems is particularly challenging. 

Kazantzis--Kravaris/Luenberger (KKL) observers \cite{kazantzis1998, andrieu2006} offer a principled approach by constructing an injective transformation that lifts the nonlinear system into a higher-dimensional space where the observer dynamics are quasi-linear and globally stable. 
Under mild observability conditions, this transformation is theoretically guaranteed to exist, and the state estimate, obtained by inverting the transformation at the latent observer state, globally asymptotically converges to the true state \cite{andrieu2006, brivadis2023, bernard2022}. 
In practice, however, synthesizing a KKL observer requires finding a transformation that satisfies a certain PDE, which turns out to be analytically intractable for general nonlinear systems. Computing a left inverse of this transformation is strictly harder, as it entails inverting a nonlinear map for which no closed-form expression is typically available \cite{andrieu2021}.

Furthermore, real-world systems are typically controlled with time-varying signals. Robotic platforms receive motor commands, traffic networks experience a time-varying demand, and industrial processes are subject to time-varying parameters \cite{astrom2021}. 
These time-varying elements fundamentally alter the state evolution and render the system non-autonomous. 
The theoretical extension of KKL observers to non-autonomous settings is established in \cite{bernard2017, bernard2018, bernard2019}, where it is shown that the transformation map must be time-varying to accommodate system inputs. 
Yet, existing learning-based methods for KKL observer design remain limited to the autonomous case, and extending them to handle inputs remains an open problem \cite{buisson2023, niazi2025}. 

\subsection{Literature Gap}

The practical bottleneck of KKL observers is the numerical approximation of the transformation map and its left inverse. 
For autonomous systems, \cite{ramos2020} and \cite{janny2021} propose supervised neural network approaches using simulated trajectory data, while \cite{buisson2023} introduces an unsupervised autoencoder architecture with PDE regularization. 
Subsequently, \cite{niazi2023, niazi2025} develop a supervised physics-informed learning method that embeds the PDE constraint directly into the training, yielding improved generalization and robustness guarantees. 
Recently, \cite{miao2023} proposes a neural ODE method and \cite{marani2025} proposes an approach based on contraction analysis. 
The non-autonomous KKL problem was analyzed from a theoretical perspective in \cite{bernard2017, bernard2018, bernard2019}. However, existing numerical synthesis methods for KKL observers for non-autonomous systems rely on static transformation mappings \cite{peralez2021, peralez2022, peralez2024, marani2025cdc, woelk2026}. These approaches extend the autonomous framework but do not incorporate time-varying transformation maps, which limits their scope to systems that satisfy the restrictive instantaneous observability condition \cite{bernard2018}. To our knowledge, no systematic study exists that compares input conditioning strategies for learning-based KKL observers under mild observability conditions \cite{bernard2017, bernard2018, bernard2019}, nor has dynamic conditioning of the transformation maps been explored.

\subsection{Our Contributions}

To address the gap above, we propose HyperKKL, a hypernetwork-based framework \cite{ha2016, chauhan2024} for conditioning KKL observers on exogenous inputs. We investigate two complementary strategies: (i) HyperKKL\textsubscript{obs}: an augmented observer that uses a hypernetwork to inject input-dependent corrections into the latent observer dynamics while retaining static transformation maps, and (ii) HyperKKL\textsubscript{dyn}: a dynamic approach that uses a hypernetwork to generate input-dependent encoder and decoder weights, producing time-varying transformation and inverse maps. 

Through this investigation, we address two fundamental questions: Q1) Is input conditioning necessary for designing KKL observers for controlled systems, and if so, which conditioning strategy is more effective: augmenting the latent dynamics or dynamically adapting the transformation maps? Q2) Why are hypernetworks the appropriate architectural mechanism for implementing input conditioning in KKL observers?

Our contributions are threefold. First, we implement and systematically compare two complementary input-conditioning strategies for learning-based neural KKL observers for non-autonomous systems: latent-dynamics augmentation and dynamic transformation conditioning. 
Second, we introduce a hypernetwork-based realization of time-dependent transformation maps that enables KKL observer design under mild observability conditions, without relying on the restrictive instantaneous observability assumption. 
Third, we derive theoretical guarantees in the form of worst-case estimation error bounds that characterize how the network's approximation capacity, the residual violations of the governing PDE, and the sensitivity of the decoder jointly shape the resulting state estimation accuracy.






\section{Problem Formulation}
\label{sec:problem}

Consider a controlled (non-autonomous) nonlinear system in continuous time $t\in\bb R_+$
\begin{subequations}
    \label{eq:sys}
    \begin{align}
        \dot x(t) &= f(x(t), u(t)) 
        \label{eq:sys-state} \\ 
        y(t) &= h(x(t)) 
        \label{eq:sys-output}
    \end{align}
\end{subequations}
where $x(t) \in \ms X \subseteq \bb R^{n_x}$ is the state, $u(t) \in \ms U \subset \bb R^{n_u}$ is a known exogenous input, and $y(t) \in \bb R^{n_y}$ is the measured output.
The input trajectory $u:\bb R_+ \to \ms U$ belongs to a set of admissible inputs $\mf U \subset \mc L_\text{loc}^\infty(\bb R_+)$, where $\mc L_\text{loc}^\infty(\bb R_+)$ denotes the Lebesgue space of \textit{locally essentially bounded measurable functions} defined on the non-negative real line $\bb R_+$. That is, the mapping $t\mapsto u(t)$ is Lebesgue measurable and, for every compact interval $\ms T \subset \bb R_+$, the function is bounded almost everywhere, i.e., $\mathrm{ess} \sup_{t \in \ms T} \|u(t)\|_\infty < \infty$. 

Let $X(t, u; \vartheta, t_0)$ denote the solution of \eqref{eq:sys-state} at time $t\in\bb R_+$, driven by the input trajectory $u\in\mf U$ and initialized from the state $x(t_0)=\vartheta\in\ms X$ at time $t_0\in\bb R_+$. Similarly, $Y(t, u; \vartheta, t_0)$ denote the output trajectory corresponding to the solution $X(t, u; \vartheta, t_0)$.
From \cite{bernard2018}, the following assumption ensures the theoretical existence of a KKL observer.

\begin{assume}
    \label{assume:system}
    System \eqref{eq:sys} satisfies the following:
    \begin{enumerate}
        \item $f$ is continuously differentiable and $h$ is continuous. 
        \item System \eqref{eq:sys-state} is \textit{forward complete} on $\ms X$ from $\ms X_0\subseteq \ms X$. That is, $\forall u\in \mf U$ and $\forall \vartheta \in\ms X_0$, the maximal \textit{interval of existence} includes $[0, +\infty)$, and the unique, \textit{absolutely continuous} solution $X(t, u; \vartheta, 0) \in \ms X$, $\forall t \in \bb R_+$.
        \item Solutions to \eqref{eq:sys-state} initiated from $\ms X$ at any time $t \in \bb R_+$ do not exhibit finite-time escape in backward time over the interval $[0, t]$. That is, for all $u\in \mf U$, for all $t \in \bb R_+$, and for all $\vartheta \in \ms X$, the unique, absolutely continuous solution $X(s, u; \vartheta, t)\not\to\infty$, $\forall s\in [0, t]$.
        \item The system is \textit{backward-distinguishable}. That is, for any given input $u\in\mf U$, there exists a finite time $t_u \ge 0$ such that $\forall t\ge t_u$ and $\forall (\vartheta_1, \vartheta_2) \in \ms X^2$, the following implication holds: if $Y(s,u; \vartheta_1,t) = Y(s,u; \vartheta_2,t)$, $\forall s\in[t-t_u,t]$, then it must be that $\vartheta_1=\vartheta_2$.
        \hfill $\diamond$
    \end{enumerate}
\end{assume}

From \cite{bernard2018}, designing a KKL observer requires a state transformation map $\mc T: \bb R^{n_x} \times \bb R_+ \to \bb R^{n_z}$ that transforms the nonlinear dynamics \eqref{eq:sys-state} into a globally stable, quasi-linear latent system
\begin{equation}
    \label{eq:latent-z}
    \dot z(t) = Az(t) + Bh(x(t))
\end{equation}
where $z(t) = \mc T(x(t), t) \in \bb R^{n_z}$ with $n_z = n_y(2n_x + 1)$, $A\in \bb R^{n_z\times n_z}$ is a Hurwitz matrix, and $B\in \bb R^{n_z\times n_y}$ is such that $(A, B)$ is a controllable pair.
Because of \eqref{eq:latent-z}, the transformation map $\mc T$ must be continuously differentiable and satisfy the following PDE
\begin{equation}
    \label{eq:pde}
    \frac{\partial \mc T}{\partial x}(x, t) f(x, u(t)) + \frac{\partial \mc T}{\partial t}(x, t) = A\mc T(x, t) + Bh(x)
\end{equation}
for all $x\in \ms X$, $t\in \bb R_+$, and $u\in \mf U$.
In addition, $\mc T$ must be injective, i.e., there exists a left inverse $\mc T^*$ such that $\mc T^*(\mc T(x,t), t)=x$ for all $x\in\ms X$ and $t\in \bb R_+$.

Unlike the autonomous case \cite{andrieu2006, brivadis2023}, where the transformation map is time-invariant, the presence of input $u$ necessitates that $\mc T(x, t)$ is time-varying and satisfies PDE \eqref{eq:pde} with temporal drift $\partial \mc T/\partial t$, which captures the non-stationary geometry of the system's attractor induced by the input $u(t)$.

Note that solving \eqref{eq:pde} analytically for every input trajectory $u\in \mf U$ and obtaining the inverse of its solution at every time $t\in \bb R_+$ are intractable tasks. Therefore, we aim to learn the time-varying maps $\mc T(x,t), \mc T^*(z,t)$ using neural networks and synthetic data generated by simulating \eqref{eq:sys}.

Let $\hat{\mc T}(x,t)$ and $\hat{\mc T}^*(z,t)$ denote the learned transformation and left inverse maps. A general form of the neural KKL observer is
\begin{subequations}
    \label{eq:kkl}
    \begin{align}
        \dot{\hat z}(t) &= A\hat z(t) + B y(t) + \Phi(\hat z(t), u(t)) \label{eq:kkl-z} \\
        \hat x(t) &= \hat{\mc T}^*(\hat z(t),t)
        \label{eq:kkl-estimate}
    \end{align}
\end{subequations}
where $\hat x(t)$ is the state estimate and $\Phi: \bb R^{n_z}\times \bb R^{n_u}\to\bb R^{n_z}$ is an optional input injection term. When the transformation maps are time-varying, the injection term may be set to $\Phi \equiv 0$, since the input dependence is absorbed into $\hat{\mc T}$ and $\hat{\mc T}^*$. Conversely, when $\hat{\mc T}$ and $\hat{\mc T}^*$ are static (time-invariant), the term $\Phi$ must compensate for the model mismatch induced by the input.
This formulation naturally leads to two complementary conditioning strategies, which we develop in~\cref{sec:framework}: $\text{HyperKKL}_\text{obs}$, which retains static maps and learns a nonzero $\Phi$, and $\text{HyperKKL}_\text{dyn}$, which learns time-varying maps and sets $\Phi\equiv 0$.

Our objective is to develop a learning-based method for obtaining the approximations $\hat{\mc T}$ and $\hat{\mc T}^*$ for the transformation $\mc T$ and its left inverse $\mc T^*$, respectively, alongside the injection term $\Phi$ when applicable, and provide a bound on the estimation error $\xi(t) = x(t) - \hat x(t)$ of the learned neural KKL observer.

\section{HyperKKL Framework}
\label{sec:framework}

Let the transformation map $\mc T$ and its left inverse $\mc T^*$ be approximated using neural networks with weights $\theta(t)$ and $\eta(t)$, i.e.,
\begin{equation}
    \label{eq:approximations}
    \hat{\mc T}(x,t) \triangleq \mc N_\text{enc}(x; \theta(t)), \quad \hat{\mc T}^*(z,t) \triangleq \mc N_\text{dec}(z; \eta(t)).
\end{equation}
We investigate two strategies for conditioning KKL observers on exogenous inputs $u(t)$: \textit{(i)}~$\text{HyperKKL}_\text{obs}$, an augmented observer that employs input injection into the latent dynamics \eqref{eq:kkl-z} while keeping the learned transformation maps and its inverse static ($\theta(t)=\ol\theta$, $\eta(t)=\ol\eta$), and \textit{(ii)}~$\text{HyperKKL}_\text{dyn}$, a dynamic transformation approach that uses a hypernetwork to generate input-dependent encoder and decoder weights $\wt\theta(t)$ and $\wt\eta(t)$, producing time-varying encoder and decoder networks with weights $\theta(t)=\ol\theta+\wt\theta(t)$ and $\eta(t)=\ol\eta+\wt\eta(t)$. Both strategies share a common first phase in which base encoder and decoder networks are pretrained on the autonomous dynamics ($u\equiv 0$) to learn the base weights $\ol\theta$ and $\ol\eta$, using the method proposed by \cite{niazi2023}. 
The Phase~1 pretraining procedure is provided in \cref{alg:phase1}.

\begin{algorithm}[!t]
\caption{Autonomous KKL training (Phase~1)}
\label{alg:phase1}
\begin{algorithmic}[1]
 
\STATE \textbf{Stage A}: Train base encoder $\mathcal{T}_{\ol \theta}\colon x \mapsto z$
\STATE Collect dataset $\mathcal{D}_\text{train-aut}=\{(x_k^i,\,z_k^i,\,y_k^i)\}$
\FOR{$\text{epoch} = 1, \dots, E_1$}
  \FOR{each mini-batch $(\bm{x}, \bm{z}, y)$ from $\mathcal{D}_\text{train-aut}$}
    \STATE $\hat{\bm{z}} \leftarrow \mathcal{T}_{\ol\theta}(\bm{x})$
    \STATE $\mathcal{L}_{\text{mse}} \leftarrow \|\hat{\bm{z}} - \bm{z}\|^2$
    \STATE $\mathcal{L}_{\text{pde}} \leftarrow \bigl\|\frac{\partial \mathcal{T}_{\ol\theta}}{\partial x}\,f(\bm{x}) - A\,\hat{\bm{z}} - B\,y\bigr\|^2$ \hfill $\triangleright$ \textit{auto-diff}
    \STATE $\nu \leftarrow \nu_{\max}\cdot\min(1,\;\text{epoch}/5)$ \hfill $\triangleright$ \textit{linear warmup}
    \STATE $\mathcal{L} \leftarrow \mathcal{L}_{\text{mse}} + \nu\mathcal{L}_{\text{pde}}$
    \STATE Update $\ol\theta$ on $\nabla_{\ol\theta} \mathcal{L}$ \hfill $\triangleright$ grad clipped to $\|\nabla\| \le 1$
  \ENDFOR
  \STATE Scheduler step (ReduceLROnPlateau)
\ENDFOR
 
\medskip
\STATE \textbf{Stage B}: Train base decoder $\mathcal{T}^*_{\ol\eta}\colon z \mapsto x$
\STATE Generate dataset $\;\mathcal{D}^*_\text{train-aut} = \{(x_k^i,\;\mathcal{T}_{\ol\theta}(x_k^i))\}$
\FOR{$\text{epoch} = 1, \dots, E_2$}
  \FOR{each mini-batch $(\bm{x}, \bm{z})$ from $\mathcal{D}^*_\text{train-aut}$}
    \STATE $\hat{\bm{x}} \leftarrow \mathcal{T}^*_{\ol\eta}(\bm{z})$
    \STATE $\mathcal{L} \leftarrow \|\hat{\bm{x}} - \bm{x}\|^2$
    \STATE Update $\ol\eta$ on $\nabla_{\ol\eta} \mathcal{L}$ \hfill $\triangleright$ grad clipped to $\|\nabla\| \le 1$
  \ENDFOR
  \STATE Scheduler step (ReduceLROnPlateau)
\ENDFOR
\end{algorithmic}
\end{algorithm}

\subsection{Data Generation and the Operating Domain}
\label{sec:data}

Both $\text{HyperKKL}_\text{obs}$ and $\text{HyperKKL}_\text{dyn}$ share the same offline dataset. Let $\ms X_0 \subset \ms X$ be a set of initial conditions and $\mf U_\text{train} \subset \mf U$ be a representative set of input signals. We sample $N_\text{traj}$ initial conditions $\{x_0^{(i)}\}_{i=1}^{N_\text{traj}} \subset \ms X_0$ and a distinct set of $N_\text{inp}$ input signals $\{u^{(j)}(\cdot)\}_{j=1}^{N_\text{inp}} \subset \mf U_\text{train}$. We simulate the system dynamics in \eqref{eq:sys-state} for every pair of initial condition and input signal over a horizon $T>0$ with a constant sampling period $\Delta t$. Let the discrete time instances be defined as $t_k = k \Delta t$, yielding a total of $N_\text{step} = T / \Delta t$ steps per trajectory. 

To accommodate the historical input window of length $\omega$ and the forward finite difference step at $t_{k+1}$, we define the valid starting index $N_\text{wind} = \lceil \omega / \Delta t \rceil$. This yields the following empirical dataset
\begin{equation}
    \label{eq:dataset}
    \mathcal{D}_\text{train} = \bigcup_{i=1}^{N_\text{traj}} \bigcup_{j=1}^{N_\text{inp}} \bigcup_{k=N_\text{wind}}^{N_\text{step}-1} \!\!\!\left\{ \!\!\left(x^{(i,j)}_k, u^{(j)}_{[t_k-\omega, t_k]}, u^{(j)}_{[t_{k+1}-\omega, t_{k+1}]} \right) \!\!\right\}
\end{equation}
where $x^{(i,j)}_k \in \mathcal{X}$ is the sampled state at time $t_k$ originating from the $i$-th initial condition and driven by the $j$-th input signal, and the adjacent input windows are stored to approximate the temporal derivative.

\subsection{$\text{HyperKKL}_\text{obs}$: Augmented Observer}
\label{sec:augmented}

The first strategy, $\text{HyperKKL}_\text{obs}$, retains the static base encoder $\hat{\mc T}(x) = \mc N_\text{enc}(x; \ol\theta)$ and decoder $\hat{\mc T}^*(z) = \mc N_\text{dec}(z; \ol\eta)$ pretrained in Phase~1 and compensates for the input dependence by augmenting the latent observer dynamics with an input injection term. Specifically, the observer takes the form
\begin{subequations}
    \label{eq:augmented-obs}
    \begin{align}
        \dot{\hat z}(t) &= A\hat z(t) + By(t) + \Phi_\varpi(\hat z(t), u_{[t-\omega, t]}) \label{eq:augmented-obs-z} \\
        \hat x(t) &= \mc N_\text{dec}(\hat z(t);\, \ol\eta)
        \label{eq:augmented-obs-x}
    \end{align}
\end{subequations}
where $\Phi_\varpi: \bb R^{n_z}\times \bb R^{\omega \times n_u}\to \bb R^{n_z}$ is the input injection map, and $\omega\in\bb N$ is the length of the input history window.

The injection map $\Phi_\varpi$ is constructed in two stages.  First, a shared gated recurrent unit (GRU)~\cite{cho2014} processes the input window $u_{[t-\omega, t]}$ and produces a temporal context embedding $h_t$, which is projected to a latent vector $\ell = W_\text{proj}h_t$ via a linear layer without bias. Second, a multi-layer perceptron (MLP) $\phi_\varpi$ takes the concatenation $[\hat z;\, \ell]$ as input and outputs the injection signal $\Phi_\varpi(\hat z, u_{[t-\omega,t]}) \in \bb R^{n_z}$.
The absence of bias in the projection layer guarantees that when $u\equiv 0$, the context embedding vanishes and $\Phi_\varpi \equiv 0$, thereby exactly recovering the autonomous KKL observer.

The learnable parameters in $\text{HyperKKL}_\text{obs}$ are the GRU weights $\gamma$ and the injection network weights $\varpi$. The base encoder and decoder remain frozen. Training minimizes the dynamics matching loss
\begin{equation}
    \label{eq:aug-loss}
    \mc L_\text{aug}(\gamma, \varpi) = \left\| \dot z_\text{obs} - \dot z_\text{true} \right\|^2
\end{equation}
over the dataset $\mathcal D_{\text{train}}$,
where $\dot z_\text{true} = \frac{\partial \mc N_\text{enc}}{\partial x}(\ol\theta)\, f(x_k, u(t_k))$ is obtained via the Jacobian-vector product through the frozen encoder, and $\dot z_\text{obs} = Az_k + By_k + \Phi_\varpi(z_k, u_{[t_k-\omega, t_k]})$ with $z_k = \mc N_\text{enc}(x_k; \ol\theta)$. The procedure is provided in \cref{alg:static_lstm}.

\begin{algorithm}[!t]
\caption{$\text{HyperKKL}_\text{obs}$ training (Phase 2)}
\label{alg:static_lstm}
\begin{algorithmic}[1]
\STATE Freeze $\mathcal{T}_{\bar\theta}$, $\mathcal{T}^*_{\bar\eta}$ \hfill $\triangleright$ \textit{no grad}
\STATE $\phi_\varpi, \mathcal{E}_\gamma \leftarrow$ Init injection network and GRU encoder
\medskip
\FOR{$\text{epoch} = 1, \dots, E$}
  \FOR{each mini-batch $(\bm{x}, y, \bm{u}_w, u, \dot{\bm{x}})$}
    \STATE $\bm{z} \leftarrow \mathcal{T}_{\bar\theta}(\bm{x})$ \hfill $\triangleright$ \textit{frozen encoder}
    \STATE $\bm{h}_N \leftarrow \text{GRU}(\bm{u}_w)$
    \STATE $\bm{\ell} \leftarrow W_{\text{proj}}\,\bm{h}_N$ \hfill $\triangleright$ \textit{latent embedding (no bias)}
    \STATE $\bm{\Phi} \leftarrow \phi_\varpi([\bm{z};\, \bm{\ell}])$ \hfill $\triangleright$ $\bm{\Phi} \in \mathbb{R}^{n_z}$
    \STATE $\dot{\bm{z}}_{\text{true}} \leftarrow \text{JVP}(\mathcal{T}_{\bar\theta},\,\bm{x},\,\dot{\bm{x}})$ \hfill $\triangleright$ $\frac{\partial \mathcal{T}}{\partial x}\,f(x,u)$
    \STATE $\dot{\bm{z}}_{\text{obs}} \leftarrow A\bm{z} + By + \bm{\Phi}$
    \STATE $\mathcal{L} \leftarrow \|\dot{\bm{z}}_{\text{obs}} - \dot{\bm{z}}_{\text{true}}\|^2$
    \STATE Update $(\gamma, \varpi)$ on $\nabla_{(\gamma,\varpi)} \mathcal{L}$ \hfill $\triangleright$ {grad clip $\|\nabla\| \le 1$}
  \ENDFOR
  \STATE Scheduler step (cosine annealing)
\ENDFOR
\end{algorithmic}
\end{algorithm}

\begin{algorithm}[!t]
\caption{$\text{HyperKKL}_\text{dyn}$ training (Phase 2)}
\label{alg:dynamic_lora}
\begin{algorithmic}[1]
\STATE Freeze $\mathcal{T}_{\ol\theta}$, $\mathcal{T}^*_{\ol\eta}$ \hfill $\triangleright$ \textit{no grad}
\STATE $\psi \leftarrow$ Init all hyper-network parameters
\medskip
\FOR{$\text{epoch} = 1, \dots, E$}
  \FOR{each mini-batch $(\bm{x}, y, \bm{u}_w, \bm{u}_w^{-}, \dot{\bm{x}})$}
    \STATE $\bm{h} \leftarrow \text{GRU}(\bm{u}_w)$ \hfill $\triangleright$ \textit{encode input history}
    \FOR{$l = 1, \dots, L$}
      \STATE $\bm{f}_l \leftarrow \text{MLP}_{\text{bb}}([\bm{h}\,;\,\bm{e}_l])$
      \STATE $[\bm{a}_l\,\|\,\bm{b}_l] \leftarrow \text{Head}_l(\bm{f}_l)$
      \STATE $\Delta W_l \leftarrow s\cdot\underbrace{\bm{a}_l}_{d_{\text{in}}\times r}\,\underbrace{\bm{b}_l}_{r\times d_{\text{out}}}$ \hfill $\triangleright$ \textit{low-rank delta}
    \ENDFOR
    \STATE $W_l' \leftarrow W_l + \Delta W_l$ for weights; keep biases fixed
    \medskip
    \STATE $\bm{z} \leftarrow \widetilde{\mathcal{T}}(\bm{x})$, \quad $\hat{\bm{x}} \leftarrow \widetilde{\mathcal{T}}^*(\bm{z})$
    \hfill $\triangleright$ \textit{modulated maps}
    \medskip
    \STATE \textit{// PDE loss}
    \STATE $\partial_t \widetilde{\mathcal{T}} \leftarrow \text{JVP}(\widetilde{\mathcal{T}}(\bm{x};\,\bm{u}_w),\;\bm{u}_w,\;\dot{\bm{u}}_w)$ 
    \STATE $\partial_x \mathcal{T}\cdot f \leftarrow \text{JVP}(\widetilde{\mathcal{T}},\,\bm{x},\,\dot{\bm{x}})$ \hfill $\triangleright$ \textit{notation: $\partial_x \triangleq \frac{\partial}{\partial x}$}
    \STATE $\mathcal{L}_{\text{pde}} \leftarrow \|\partial_t \mathcal{T} + \partial_x \mathcal{T}\cdot f - A\bm{z} - By\|^2$
    \medskip
    \STATE $\mathcal{L}_{\text{recon}} \leftarrow \|\hat{\bm{x}} - \bm{x}\|^2$
    \STATE $\mathcal{L} \leftarrow \mathcal{L}_{\text{recon}} + \lambda \cdot \mathcal{L}_{\text{pde}}$
    \STATE Update $\psi$ on $\nabla_\psi \mathcal{L}$ \hfill $\triangleright$ grad clipped to $\|\nabla\| \le 1$
  \ENDFOR
  \STATE Scheduler step (cosine annealing)
\ENDFOR
\end{algorithmic}
\end{algorithm}

\subsection{$\text{HyperKKL}_\text{dyn}$: Dynamic Transformation Maps}
\label{sec:dynamic}

\begin{figure}[!t]
    \centering
    \includegraphics[width=\columnwidth]{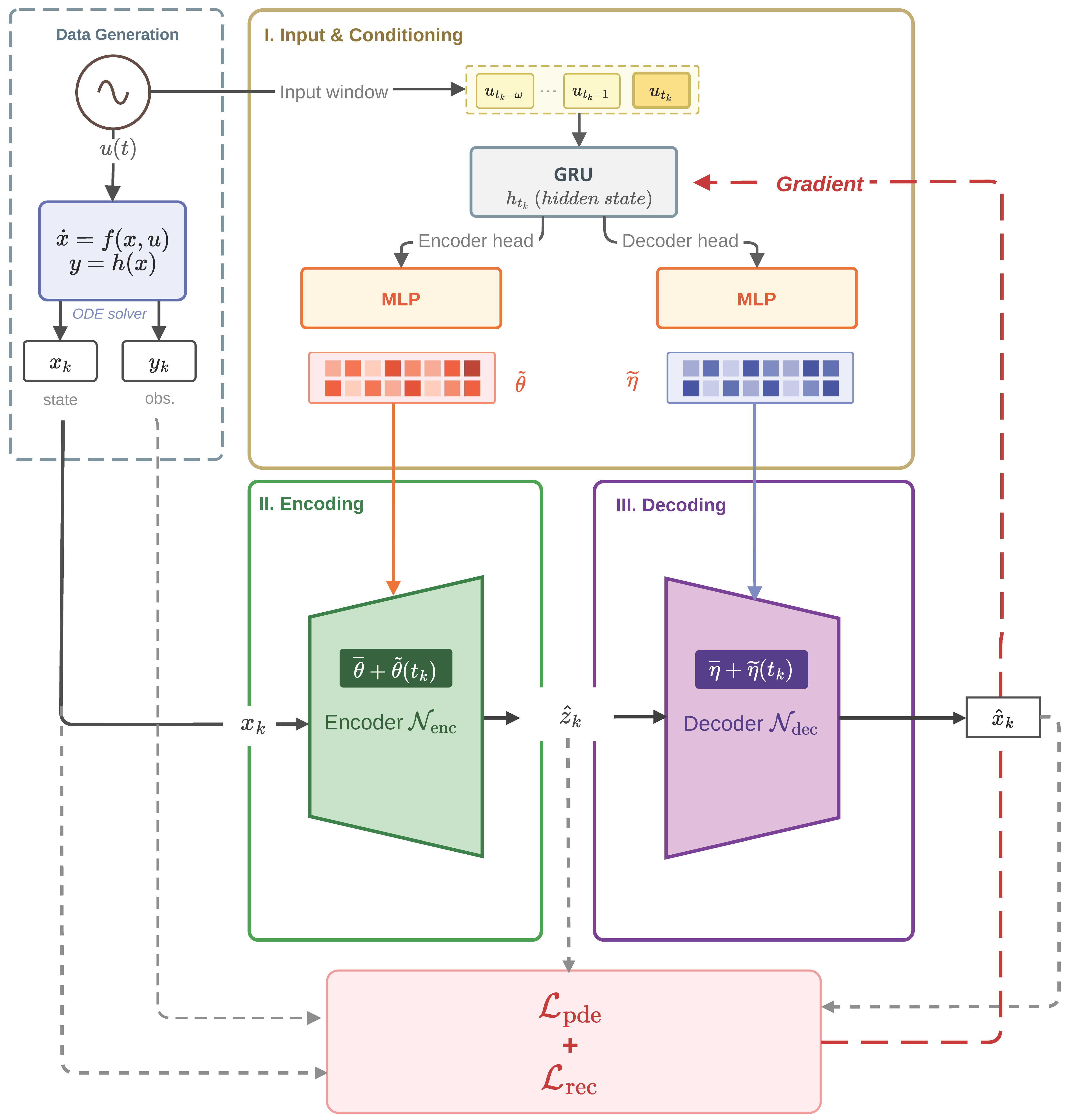}
    \caption{$\text{HyperKKL}_\text{dyn}$ training architecture. The shared GRU processes the input window $u_{[t_k-\omega, t_k]}$ to produce a context embedding $h_t$, from which separate MLP heads predict weight perturbations $\wt\theta(t_k)$ and $\wt\eta(t_k)$. These are added to the frozen base weights $\ol\theta$ and $\ol\eta$, yielding an input-conditioned encoder and decoder. The losses $\mc L_\text{pde}$ and $\mc L_\text{rec}$ are backpropagated only through the hypernetwork parameters $\psi$. }
    \label{fig:architecture}
\end{figure}

The second strategy, $\text{HyperKKL}_\text{dyn}$, directly addresses the time-varying nature of the transformation map by conditioning the encoder and decoder weights on the input. 
The parameters of the networks in \eqref{eq:approximations} are decomposed into the static base weights $\ol\theta,\ol\eta$ (pretrained in Phase~1) and dynamic perturbations $\wt\theta(t),\wt\eta(t)$ as
\begin{equation}
    \label{eq:parameter-decomposition}
    \theta(t) = \ol\theta + \wt\theta(t), \quad \eta(t) = \ol\eta + \wt\eta(t).
\end{equation}
We impose a hypernetwork ansatz where the perturbations are functions of the input history over a window $\omega\in\mathbb N$, i.e.,
\begin{equation}
    \label{eq:hyper-net}
    [\wt\theta(t),\wt\eta(t)] = \mc H_\psi(u_{[t-\omega, t]})
\end{equation}
where the hypernetwork $\mc H_\psi$ is parameterized by $\psi$, which is learnable. Here, the hypernetwork $\mc H_\psi$ consists of a shared GRU to extract a temporal context embedding, followed by separate MLPs to predict the parameter shifts for the encoder $\mc N_\text{enc}(x; \theta(t))$ and decoder $\mc N_\text{dec}(z; \eta(t))$. 

Since the transformation maps are now time-varying by construction, the latent observer dynamics do not require an explicit injection term. The observer takes the form
\begin{subequations}
    \label{eq:dynamic-obs}
    \begin{align}
        \dot{\hat z}(t) &= A\hat z(t) + By(t) \label{eq:dynamic-obs-z} \\
        \hat x(t) &= \mc N_\text{dec}(\hat z(t);\, \eta(t))
        \label{eq:dynamic-obs-x}
    \end{align}
\end{subequations}
where $\eta(t) = \ol\eta + \wt\eta(t)$ is generated by the hypernetwork at each time step.
The architecture is illustrated in \cref{fig:architecture}.


We define two loss components evaluated over $\ms D_\text{train}$ for $\text{HyperKKL}_\text{dyn}$.
First, the reconstruction loss is given by 
\begin{equation}
    \label{eq:recon-loss}
    \mc L_\text{rec}(\psi) \!=\! \left\| x_k - \mc N_\text{dec} \Big( \mc N_\text{enc}(x_k; \theta(t_k; \psi)); \eta(t_k; \psi) \Big) \right\|^2\!.
\end{equation}
Second, the dynamic PDE residual is given by
\begin{multline}
    \label{eq:pde-loss}
    \mc L_\text{pde}(\psi) = \bigg\| \frac{\partial \mc N_\text{enc}}{\partial x}(x_k; \theta(t_k; \psi)) f(x_k, u(t_k)) \\ 
    + \frac{\partial \hat{\mc T}}{\partial t}(x_k,t_k; \psi) \!-\! \Big(\! A \mc N_\text{enc}(x_k; \theta(t_k; \psi)) \!+\! B h(x_k) \!\Big) \bigg\|^2\!.
\end{multline}
The learning problem is then formulated as empirical risk minimization over the dataset $\ms D_\text{train}$ as
\begin{equation}
    \label{eq:learning-prob}
    \min_{\psi} \frac{1}{|\ms D_\text{train}|} \sum_{\ms D_\text{train}} \Big( \mc L_\text{rec}(\psi) + \nu \mc L_\text{pde}(\psi) \Big)
\end{equation}
where $\nu > 0$ is a hyperparameter that balances the topological consistency of the latent space against the algebraic accuracy of the state reconstruction. 
The procedure is provided in \cref{alg:dynamic_lora}.

\subsection{Comparison of $\text{HyperKKL}_\text{obs}$ and $\text{HyperKKL}_\text{dyn}$}

$\text{HyperKKL}_\text{obs}$ and $\text{HyperKKL}_\text{dyn}$ represent fundamentally different approaches to input conditioning. $\text{HyperKKL}_\text{obs}$ keeps the static transformation maps fixed after learning and compensates for the input in the latent dynamics via the injection term $\Phi_\varpi$. This is computationally lighter, as only the injection network and GRU encoder are trained, and the observer retains the standard quasi-linear structure in the latent space. However, because the encoder and decoder are static, this strategy implicitly assumes that the autonomous transformation provides a sufficiently accurate embedding for the input-driven regime, with the injection term correcting only the residual dynamics mismatch.

$\text{HyperKKL}_\text{dyn}$, by contrast, directly modulates the parameters of the learned encoder and decoder to track the time-varying geometry of the system's attractor, as prescribed by the theory in \cite{bernard2018}. 
The hypernetwork generates input-dependent perturbations to both the encoder and decoder weights, which enables the observer to adapt the transformation and its inverse to each input trajectory realization. 
This is more expressive but also more computationally demanding, as it requires backpropagation through the generated weights.

\section{Theoretical Analysis and Design Trade-offs}
\label{sec:theory}

In this section, we provide a theoretical bound on the estimation error of a HyperKKL observer. For $x\in \ms X$ and $u\in \ms U$, define the PDE residual as
\begin{equation}
    \label{eq:pde-residual}
    \mc R_\text{pde}(x,u,t) \!=\! \frac{\partial\hat{\mc T}}{\partial x}f(x, u) + \frac{\partial\hat{\mc T}}{\partial t} - (A\hat{\mc T}(x, t) + Bh(x)).
\end{equation}
Define $\zeta(t) \triangleq \hat{\mc T}(x(t),t)$, where $x(t)$ is the true state and $\hat{\mc T}$ is the learned transformation. Then, it holds that
\begin{equation}
    \label{eq:zeta-dot}
    \dot\zeta(t) = A\zeta(t) + Bh(x(t)) + \mc R_\text{pde}(x(t),u(t),t)
\end{equation}
where $\mc R_\text{pde}(x(t),u(t),t)$ is the PDE residual in \eqref{eq:pde-residual}.
Let $\xi_z(t) = \zeta(t) - \hat z(t)$ be the latent space estimation error. From \eqref{eq:kkl-z} and \eqref{eq:zeta-dot}, it holds that
\[
\dot \xi_z(t) = A\xi_z(t) + \mc R_\text{pde}(x(t),u(t),t).
\]
Since $A$ is a Hurwitz matrix, there exist positive constants $\kappa, \lambda$ such that (see \cite[Appendix C.5]{sontag2013})
\begin{multline}
    \label{eq:xi-z-bound}
    \|\xi_z(t)\| \le \kappa e^{-\lambda t} \|\xi_z(0)\| \\
    + \int_0^t \!\kappa e^{-\lambda(t-\tau)} \|\mc R_\text{pde}(x(\tau),u(\tau),\tau)\| d\tau.
\end{multline}
Here, $\lambda>0$ is proportional to the eigenvalue of $A$ that is closest to the imaginary axis in the complex plane, and $\kappa\ge 1$ is the condition number of the eigenvector matrix of $A$.

\begin{lem}
    \label{lem:NN-bounds}
    Let the state space $\ms X$ and the set of admissible inputs $\mf U$ be compact. Assume the neural network architectures $\mc N_\text{enc}, \mc N_\text{dec}$, and $\mc H_\psi$ utilize continuously differentiable activation functions. Then, the following statements hold:
    \begin{enumerate}
        \item There exists $\epsilon_\text{pde} \ge 0$ such that the PDE residual $\mc R_\text{pde}(x, u, t)$ in \eqref{eq:pde-residual} is uniformly bounded across all states $x\in \ms X$, inputs $u\in \mf U$, and times $t\in\bb R_+$, i.e.,
        \begin{equation}
            \label{eq:Delta-pde-bound}
            \sup_{x \in \ms X, u\in \mf U, t \ge 0} \|\mc R_\text{pde}(x, u, t)\| \le \epsilon_\text{pde}.
        \end{equation}
        \item There exist $\epsilon_\text{enc} \ge 0$ and $\epsilon_\text{dec} \ge 0$ such that 
        \begin{align}
            \sup_{x \in \ms X, t \ge 0} \|\mc T(x, t) - \hat{\mc T}(x, t)\| &\le \epsilon_\text{enc} 
            \label{eq:enc-bound}
            \\
            \sup_{z \in \ms Z, t \ge 0} \|\mc T^*(z, t) - \hat{\mc T}^*(z, t)\| &\le \epsilon_\text{dec}
            \label{eq:dec-bound}
        \end{align}
        where $\ms Z \supseteq \mc T(\ms X)$ is a compact subset of $\bb R^{n_z}$ and $\hat{\mc T},\hat{\mc T}^*$ are defined in \eqref{eq:approximations}.
    \end{enumerate}
\end{lem}
\begin{proof}
    The proof follows directly from the extreme value theorem (e.g., \cite[ch.~9, p.~200]{royden2010}) and the universal approximation theorem for neural networks \cite{cybenko1989, hornik1989}.
\end{proof}

Assuming the activation functions are continuously differentiable, the learned maps are uniformly Lipschitz continuous on their domains for all $t\in \bb R_+$. The Lipschitz constant of the learned left inverse $\hat{\mc T}^*(z, t)$, denoted by $\ell_\text{dec} > 0$, is particularly relevant here. That is, for all $t\in\bb R_+$,
\begin{equation}
    \label{eq:dec-lip}
    \|\hat{\mc T}^*(z,t) - \hat{\mc T}^*(\hat z,t)\| \le \ell_\text{dec} \|z - \hat z\|, \quad \forall z,\hat z \in \ms Z.
\end{equation}

\begin{prop}
    \label{prop:error-bound}
    Let \cref{assume:system} and the conditions of \cref{lem:NN-bounds} hold. Then, given the system state $x(t)$ from \eqref{eq:sys-state} and its estimate $\hat x(t)$ from \eqref{eq:kkl-estimate}, the estimation error $\xi(t) = x(t)-\hat x(t)$ is bounded for all $t \in \bb R_+$ by
    \begin{equation}    
        \label{eq:error-bound}
        \|\xi(t)\| \!\le \epsilon_\text{dec} + \ell_\text{dec} \left( \kappa e^{-\lambda t} \|\xi_z(0)\| + \epsilon_\text{pde}\frac{ \kappa}{\lambda} + \epsilon_\text{enc}\right)
    \end{equation}
    where all the scalars above are from \eqref{eq:xi-z-bound}--\eqref{eq:dec-lip}.
\end{prop}
\begin{proof}
    By using \eqref{eq:Delta-pde-bound} to bound the right hand side of \eqref{eq:xi-z-bound} and evaluating the integral, we obtain
    \begin{equation}
        \label{eq:xi-z-1}
        \|\xi_z(t)\| \le \kappa e^{-\lambda t}\|\xi_z(0)\| + \epsilon_\text{pde}\frac{\kappa}{\lambda}.
    \end{equation}
    Recall the variables $z(t) = \mc T(x(t),t)$ and $\zeta(t) = \hat{\mc T}(x(t),t)$. Then, $x(t) = \mc T^*(z(t),t)$. We can write the estimation error $\xi(t) = \mc T^*(z(t),t) - \hat{\mc T}^*(\hat z(t),t)$.
    Then, by adding and subtracting $\hat{\mc T}^*(z(t), t)$ and $\hat{\mc T}^*(\zeta(t), t)$, taking norm, and using the triangle inequality, we obtain
    \begin{multline*}
        \|\xi(t)\| \le \|\mc T^*(z(t), t) - \hat{\mc T}^*(z(t), t)\| \\ 
        + \|\hat{\mc T}^*(z(t), t) - \hat{\mc T}^*(\zeta(t), t)\| \\
        + \|\hat{\mc T}^*(\zeta(t), t) - \hat{\mc T}^*(\hat{z}(t), t)\|.
    \end{multline*}
    From \eqref{eq:dec-bound}, the first term is bounded by $\epsilon_\text{dec}$. From \eqref{eq:dec-lip} and \eqref{eq:enc-bound}, the second term is bounded as
    \[
    \|\hat{\mc T}^*(z(t), t) - \hat{\mc T}^*(\zeta(t), t)\| \le \ell_\text{dec} \epsilon_\text{enc}.
    \]
    Finally, the third term is bounded as
    \[
    \|\hat{\mc T}^*(\zeta(t), t) - \hat{\mc T}^*(\hat{z}(t), t)\| \le \ell_\text{dec} \|\xi_z(t)\|.
    \]
    By using \eqref{eq:xi-z-1} above, we arrive at \eqref{eq:error-bound}.
\end{proof}

The exponential term decays to zero, so the asymptotic bound is
\begin{equation}
    \label{eq:asymptotic-bound}
    \limsup_{t\to\infty} \|\xi(t)\| \!\le \epsilon_\text{dec} + \ell_\text{dec} \left(\epsilon_\text{pde}\frac{ \kappa}{\lambda} + \epsilon_\text{enc}\right).
\end{equation}
This bound applies to both variants but they affect its terms differently. $\text{HyperKKL}_\text{obs}$ freezes the decoder, preserving $\ell_\text{dec}$ at its pretrained value, but $\epsilon_\text{pde}$ may be larger since static maps cannot exactly satisfy the time-varying PDE~\eqref{eq:pde}. $\text{HyperKKL}_\text{dyn}$ can reduce $\epsilon_\text{pde}$ by adapting the encoder, but the weight perturbations $\wt\eta(t)$ may inflate $\ell_\text{dec}$. Since $\ell_\text{dec}$ multiplies both $\epsilon_\text{pde}$ and $\epsilon_\text{enc}$, even moderate increases can offset gains from a smaller PDE residual. This bias--variance trade-off explains why $\text{HyperKKL}_\text{obs}$ often outperforms $\text{HyperKKL}_\text{dyn}$ empirically (\cref{sec:experiments}).

The bound suggests several practical guidelines. The constants $\epsilon_\text{enc}$, $\epsilon_\text{dec}$ are reduced by increasing network expressivity and ensuring dense sampling of $\ms X_0$. The PDE residual $\epsilon_\text{pde}$ is controlled by the weight $\nu$ in \eqref{eq:learning-prob} and reduced by decreasing $\Delta t$. The Lipschitz constant $\ell_\text{dec}$ plays the most significant role as it amplifies both error sources; we apply spectral normalization \cite{miyato2018} to the decoder to bound it. Finally, $\kappa$ and $\lambda$ are set by $A$: we choose $A$ diagonal to ensure $\kappa = 1$ and place eigenvalues to balance convergence rate against numerical stiffness.

\begin{rem}
    \label{rem:noise}
    In practical applications, system \eqref{eq:sys} could be influenced by process noise $w$ and measurement noise $v$, i.e.,
    \[
    \dot x(t) = f(x(t),u(t)) + w(t), \quad y(t) = h(x(t)) + v(t).
    \]
    Assume that for every compact interval $\ms T\subset\bb R_+$
    \[
    \mathrm{ess}\sup_{t\in\ms T}\|w(t)\|_\infty \le \bar w, \quad \mathrm{ess}\sup_{t\in\ms T}\|v(t)\|_\infty \le \bar v
    \]
    for some known $\bar w\ge 0$ and $\bar v\ge 0$. Furthermore, assuming the activation functions are continuously differentiable, implying 
    $
    \sup_{x \in \ms X, t \ge 0} \left\|\frac{\partial \hat{\mc T}}{\partial x}(x, t)\right\| \le \ell_\text{enc}
    $
    for some $\ell_\text{enc}>0$.
    Then, a straightforward corollary to \cref{prop:error-bound} is the following asymptotic bound on the estimation error
    \[
    \limsup_{t \to \infty} \|\xi(t)\| \le \epsilon_\text{dec} + \ell_\text{dec} \bigg( \epsilon_\text{pde} \frac{\kappa}{\lambda} + \epsilon_\text{enc} + \frac{\kappa}{\lambda} \ell_\text{enc} \bar w + \frac{\kappa \|B\|}{\lambda} \bar v \bigg).
    \]
    One arrives at the above bound by following similar lines of argument as in the proof of \cref{prop:error-bound}. \hfill $\diamond$
\end{rem}

\section{Numerical Validation}
\label{sec:experiments}
\begin{figure*}[!t]
    \centering
    \begin{subfigure}[t]{0.58\textwidth}
        \centering
        \includegraphics[width=\textwidth]{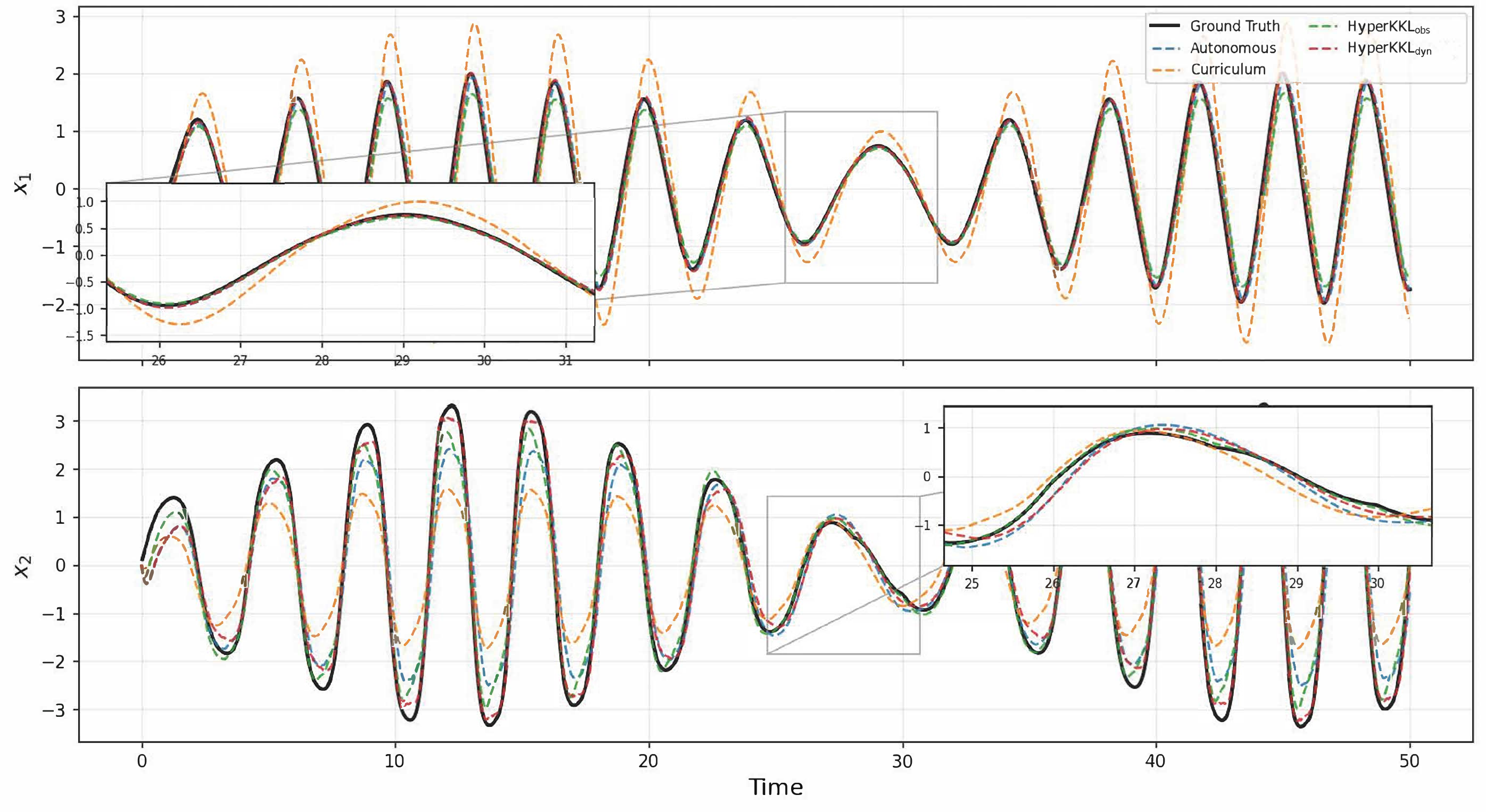}
        \caption{State estimation under square wave input for the Duffing oscillator.}
        \label{fig:timeseries-duffing}
    \end{subfigure}
    \hfill
    \begin{subfigure}[t]{0.38\textwidth}
        \centering
        \includegraphics[width=\textwidth]{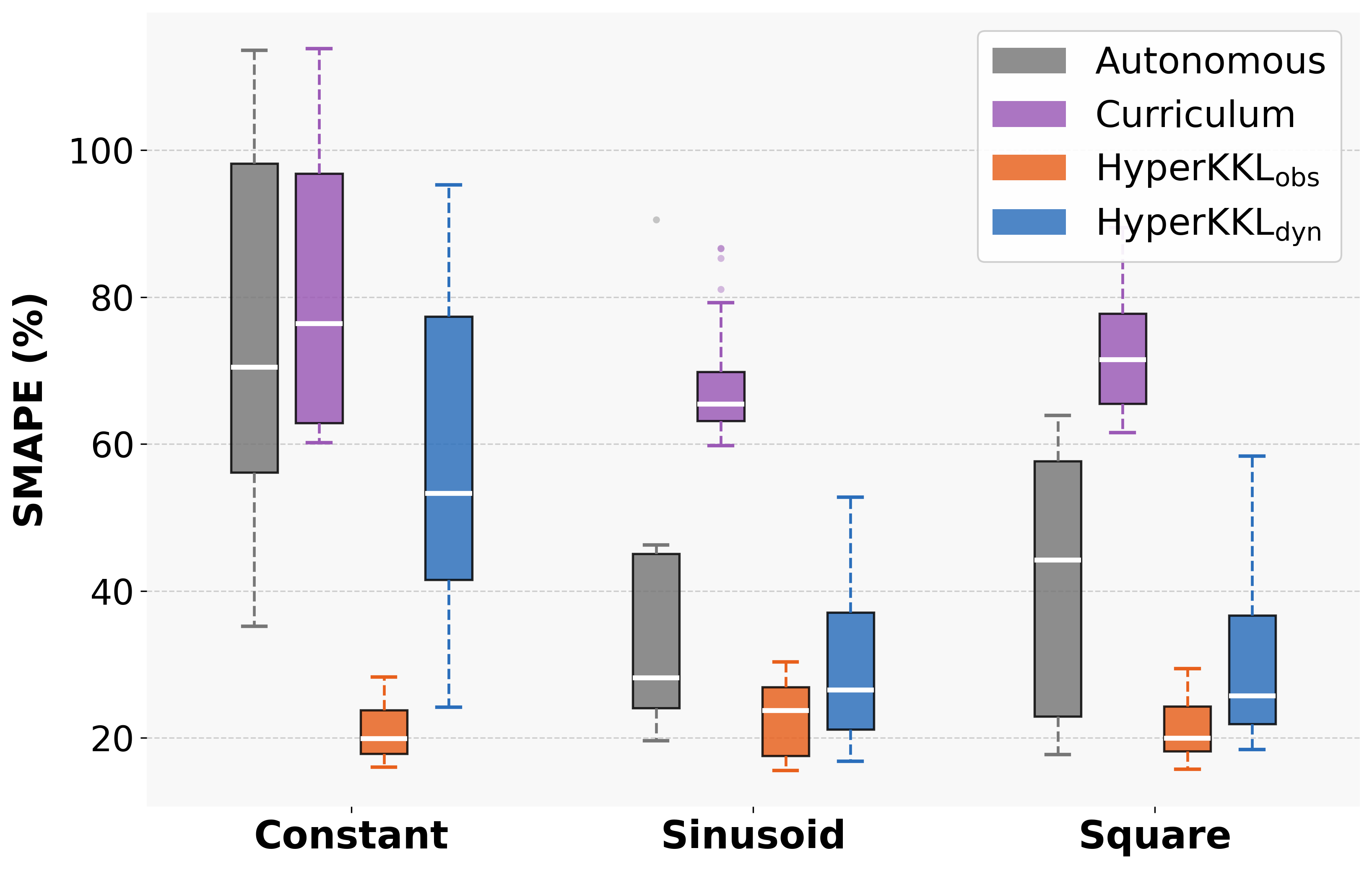}
        \caption{SMAPE: Duffing oscillator.}
        \label{fig:box-duffing}
    \end{subfigure}
 
    \vspace{6pt}
 
    \begin{subfigure}[t]{0.32\textwidth}
        \centering
        \includegraphics[width=\textwidth]{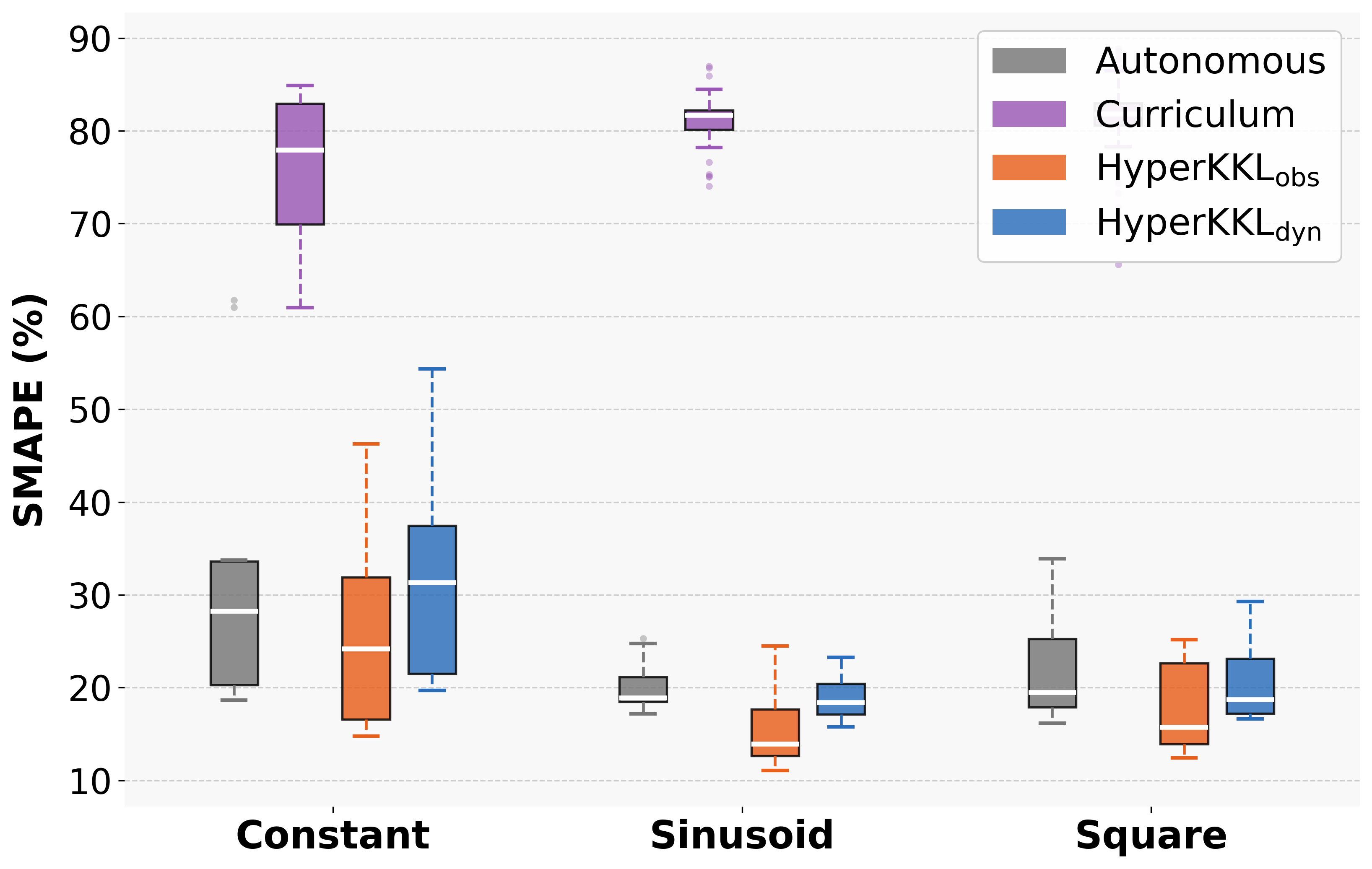}
        \caption{SMAPE: Van der Pol.}
        \label{fig:box-vdp}
    \end{subfigure}
    \hfill
    \begin{subfigure}[t]{0.32\textwidth}
        \centering
        \includegraphics[width=\textwidth]{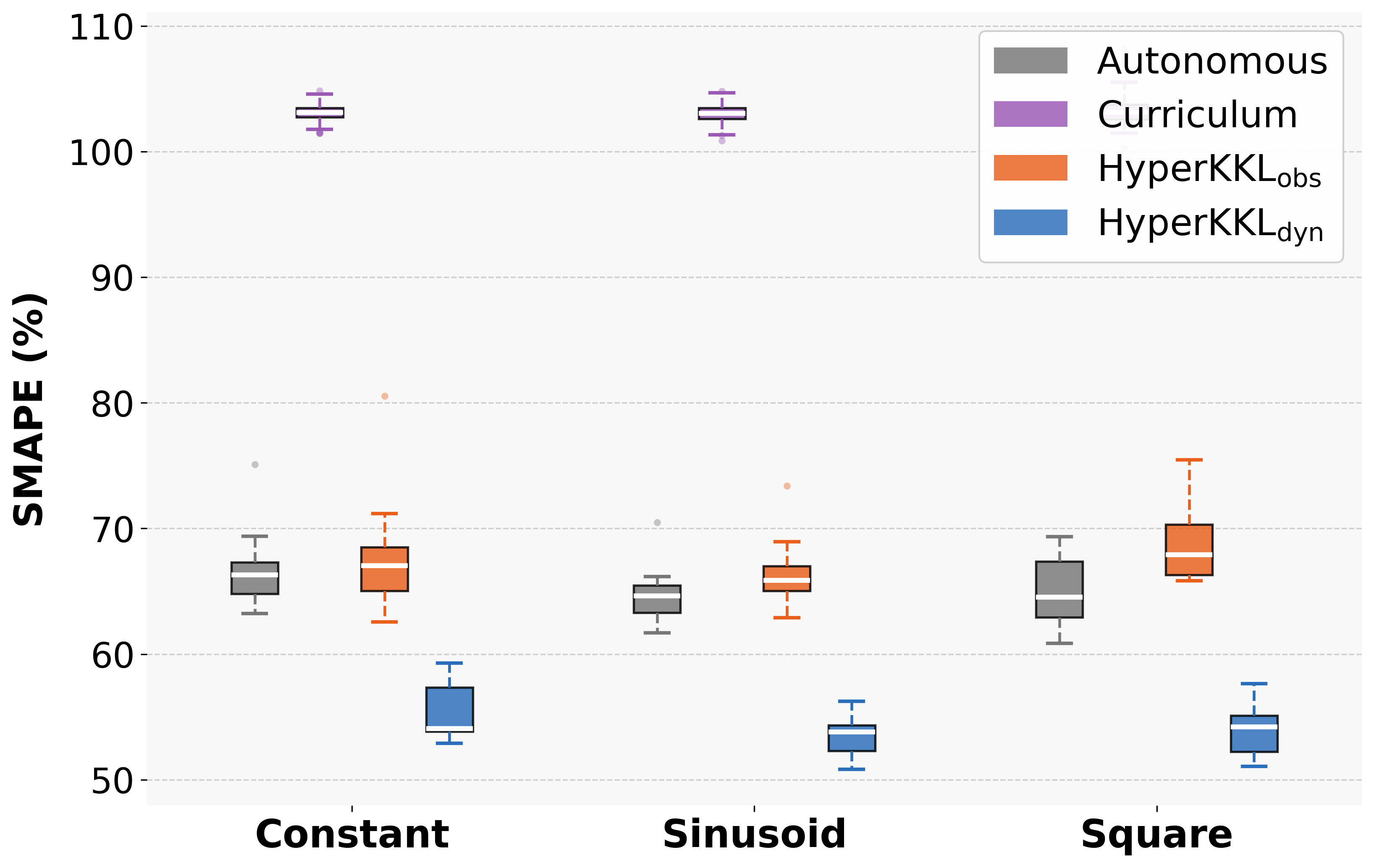}
        \caption{SMAPE: R\"{o}ssler system.}
        \label{fig:box-rossler}
    \end{subfigure}
    \hfill
    \begin{subfigure}[t]{0.32\textwidth}
        \centering
        \includegraphics[width=\textwidth]{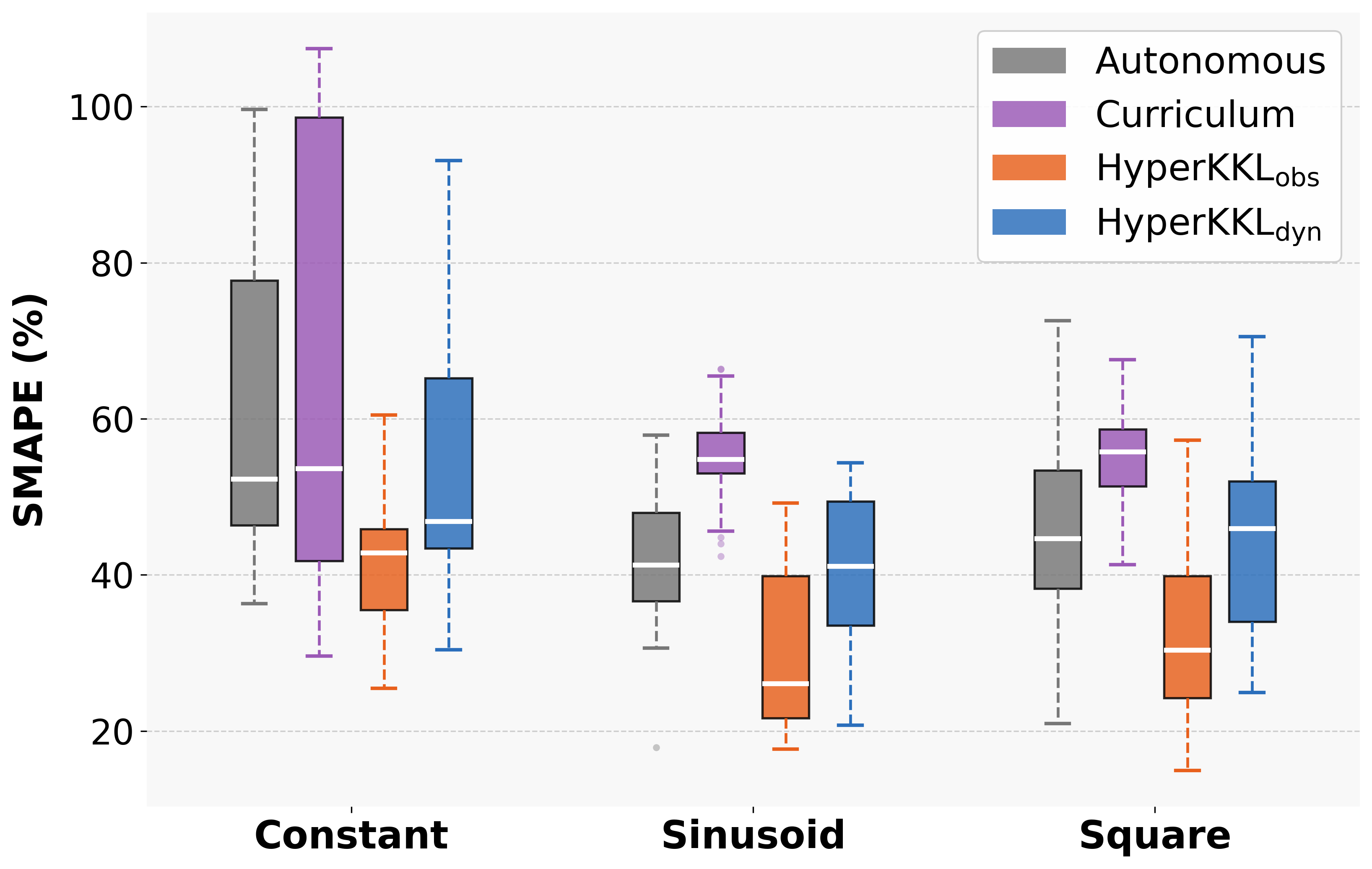}
        \caption{SMAPE: FitzHugh--Nagumo.}
        \label{fig:box-fhn}
    \end{subfigure}
 
    \caption{Estimation results across all benchmark systems over 100 trials with random initial conditions. (a)~$\text{HyperKKL}_\text{obs}$ tracks the ground truth across input transitions, while the Autonomous baseline accumulates persistent phase drift. (b)--(e)~SMAPE distributions confirm statistically robust improvements.}
    \label{fig:all-results}
\end{figure*}
 
We evaluate both HyperKKL variants on four benchmark systems. We use RK45 ODE solver with $\Delta t = 0.05$\,s over $T = 50$\,s. 
Process and measurement noise (Gaussian, $\sigma^2 = 0.01$) are added at evaluation. 
Results are averaged over 100 trials with random initial conditions, and we report the Symmetric Mean Absolute Percentage Error (SMAPE) for its scale-invariance. Code is available at \url{https://github.com/yehias21/HyperKKL}.
 
\subsection{Nonlinear Systems}
We consider the following systems for evaluation:
\begin{itemize}
    \item \textit{Duffing:} $\dot{x}_{1} = x_{2}$, $\dot{x}_{2} = -x_{1} - x_{1}^{3} + u$, $y=x_1$.
    \item \textit{Van der Pol:} $\dot{x}_{1} = x_{2}$, $\dot{x}_{2} = (1 - x_{1}^{2})x_{2} - x_{1} + u$, $y = x_1$.
    \item \textit{R\"{o}ssler:} $\dot{x}_{1} = -(x_{2} + x_{3})$, $\dot{x}_{2} = x_{1} + 0.1x_{2} + u$, $\dot{x}_{3} = 0.1 + x_{3}(x_{1} - 14)$, $y = x_2$.
    \item \textit{FitzHugh--Nagumo:} $\dot{x}_{1} = 10(x_{1} - x_{1}^{3} - x_{2}) + u$, $\dot{x}_{2} = 1.5x_{1} - x_{2} + 0.8$, $y = x_1$.
\end{itemize}
Each system is tested under four input regimes: zero ($u\equiv 0$), constant ($u\sim\mathcal{U}[-1,1]$), sinusoidal, and square wave.
 
\subsection{Baselines and Implementation}
 
We compare the proposed $\text{HyperKKL}_\text{obs}$ and $\text{HyperKKL}_\text{dyn}$ against: \textit{(i)} the \emph{autonomous observer}, which is trained on unforced dynamics using \cite{niazi2023}, and \textit{(ii)} the \emph{curriculum learning} baseline \cite{bengio2009}, which fine-tunes the pretrained maps on non-autonomous data in stages of increasing complexity, without modifying the network architecture, thereby isolating the effect of data diversity from that of architectural changes.
 
The base encoder and decoder are 3-layer MLPs (150 hidden units for Duffing and Van der Pol oscillators, 350 for R\"{o}ssler and FitzHugh-Nagumo systems). For $\text{HyperKKL}_\text{obs}$, the injection network $\phi_\varpi$ is a 2-layer MLP (64 units) mapping the GRU context to $\Phi_\varpi \in \bb R^{n_z\times n_u}$, whereas the base maps remain frozen. For $\text{HyperKKL}_\text{dyn}$, a per-layer LoRA hypernetwork uses a shared GRU (64 units, $\omega = 100$ steps) producing a context vector concatenated with per-layer embeddings (dim.\ 16), passed through a backbone MLP. Separate heads predict low-rank factors $F_\ell \in \bb R^{d_{\rm in} \times r}$, $E_\ell \in \bb R^{r \times d_{\rm out}}$ with $r = 4$, so $\Delta W_\ell = \alpha F_\ell E_\ell$ ($\alpha$ initialized at $0.01$); biases are frozen, guaranteeing exact collapse to the autonomous observer when $u \equiv 0$. Training uses Adam ($\text{lr} = 10^{-4}$) with gradient clipping. The observer matrices are $A = -\text{diag}(1, 2, \dots, n_z)$, $B = \mb 1_{n_z}$, $n_z = n_y(2n_x + 1)$. $\text{HyperKKL}_\text{obs}$ adds ${\sim}10$K parameters; $\text{HyperKKL}_\text{dyn}$ adds ${\sim}45$K; all base parameters remain frozen.
 
\subsection{Results and Discussion}

\begin{table}[h]
    \centering
    \caption{State estimation SMAPE averaged over 100 trials. Lower is better. \textbf{Bold}: best per column. $\downarrow$: improvement over Autonomous.}
    \label{tab:benchmark-results}
    \vspace{4pt}
    \newcommand{\down}{{\scriptsize$\downarrow$}}
    \newcommand{\best}[1]{\textbf{#1}}
    \resizebox{\columnwidth}{!}{
\begin{tabular}{l cccc cccc}
    \toprule
    & \multicolumn{4}{c}{\textbf{Duffing}} & \multicolumn{4}{c}{\textbf{Van der Pol}} \\
    \cmidrule(lr){2-5} \cmidrule(lr){6-9}
    \textbf{Method} & Zero & Const & Sin & Sqr & Zero & Const & Sin & Sqr \\
    \midrule
        Autonomous                     & 19.3 & 74.6 & 37.1 & 40.7 & 17.7 & 32.6 & 20.2 & 22.2 \\
        Curriculum                     & 33.0 & 80.5 & 67.3 & 71.9 & 51.4 & 76.3 & 81.2 & 81.4 \\
        \addlinespace
        \rowcolor{blue!8}
        HyperKKL$_{\text{obs}}$        & \best{5.6}\down & \best{21.1}\down & \best{22.7}\down & \best{21.3}\down & 5.3\down & \best{26.0}\down & \best{15.5}\down & \best{17.9}\down \\
        \rowcolor{blue!8}
        HyperKKL$_{\text{dyn}}$        & 8.2\down & 57.3\down & 29.7\down & 30.7\down & \best{5.0}\down & 32.7 & 18.8\down & 20.8\down \\
    \midrule\midrule
    & \multicolumn{4}{c}{\textbf{R\"{o}ssler}} & \multicolumn{4}{c}{\textbf{FitzHugh--Nagumo}} \\
    \cmidrule(lr){2-5} \cmidrule(lr){6-9}
    \textbf{Method} & Zero & Const & Sin & Sqr & Zero & Const & Sin & Sqr \\
    \midrule
        Autonomous                     & 64.9 & 66.9 & 64.8 & 64.9 & \best{9.8} & 61.4 & 41.3 & 45.0 \\
        Curriculum                     & 98 & 103.1 & 103.0 & 103.3 & 59.6 & 67.0 & 55.6 & 54.9 \\
        \addlinespace
        \rowcolor{blue!8}
        HyperKKL$_{\text{obs}}$        & 64.4\down & 68.0 & 66.5 & 68.8 & 9.9 & \best{42.1}\down & \best{30.9}\down & \best{32.5}\down \\
        \rowcolor{blue!8}
        HyperKKL$_{\text{dyn}}$        & \best{52.7}\down & \best{55.4}\down & \best{53.5}\down & \best{54.0}\down & 17.0 & 54.6\down & 40.8\down & 44.5\down \\
    \bottomrule
\end{tabular}
}
\end{table}
 
We report the SMAPE in \cref{tab:benchmark-results} and discuss the results in light of the two motivating questions posed in the introduction.
 
\emph{Q1: Is input conditioning necessary, and which strategy is more effective?}

Both $\text{HyperKKL}_\text{obs}$ and $\text{HyperKKL}_\text{dyn}$ consistently reduce SMAPE relative to the autonomous baseline under non-zero forcing. 
Under zero input, both variants recover the autonomous performance.
The boxplots for 100 trials (\cref{fig:box-duffing}--\ref{fig:box-fhn}) confirm that these gains are statistically consistent.

$\text{HyperKKL}_\text{obs}$ achieves the best or tied-best SMAPE in 13 out of 16 system/input combinations, while $\text{HyperKKL}_\text{dyn}$ is best in only 2. On the Duffing oscillator under constant forcing, $\text{HyperKKL}_\text{obs}$ reduces SMAPE from 0.569 to 0.079 (86\% reduction), compared to 0.387 for $\text{HyperKKL}_\text{dyn}$.
On the R\"{o}ssler system under sinusoidal forcing, neither variant improves over the autonomous baseline. The chaotic attractor's sensitivity, combined with the higher latent dimension, makes input conditioning harder and may require longer input windows or richer architectures.

We attribute the better performance of $\text{HyperKKL}_\text{obs}$ over $\text{HyperKKL}_\text{dyn}$ to the optimization process rather than a representational advantage. 
$\text{HyperKKL}_\text{obs}$ decouples the injection of the input signal from the modulation of the neural KKL observer. 
That is, the GRU processes the input history into a compact embedding $\ell$, and the injection network $\phi_\omega$ maps $\ell$ to a correction in the latent dynamics without interfering with the pretrained maps $\hat{\mathcal T}$ and $\hat{\mathcal T}^*$. 
Input processing and observer conditioning are thus learned independently in $\text{HyperKKL}_\text{obs}$, which yields a simpler optimization landscape. 
In $\text{HyperKKL}_\text{dyn}$, these two tasks are entangled, as the hypernetwork must simultaneously interpret the input and reshape all network weights, leading to a harder joint optimization.

The curriculum-learning baseline, i.e., the data-only paradigm, consistently degrades estimation performance. 
This demonstrates that the limitation of autonomous KKL observers under forcing is not a data problem but an architectural one. 
Exposing the static maps $\hat{\mc T}(x), \hat{\mc T}^*(z)$ to richer training data harms performance by corrupting the autonomous embedding, because the architecture itself cannot represent the time-varying solution of PDE \eqref{eq:pde}. Dynamic weights via input conditioning are a necessary architectural ingredient.

It is important to discuss why the autonomous observer still provides reasonable estimates under moderate inputs. From \cite{bernard2018}, an input $u(t)$ applied to a neural KKL observer, which was learned under $u\equiv 0$, can be interpreted as an unmodeled disturbance. The Hurwitz stability of the latent dynamics provides inherent robustness. 
Thus, a simple input induces a bounded, persistent estimation error rather than divergence. 
This explains the degraded but stable behavior of the autonomous baseline.
While the autonomous observer tolerates the input as noise, the HyperKKL variants actively incorporate it, converting an unmodeled disturbance into structured information that reduces the estimation error.

\emph{Q2: Why are hypernetworks the appropriate mechanism?}

Input conditioning in neural KKL observers enables generalization across a continuous family of input signals (constants, sinusoids, square waves) using a single trained model, without retraining for each new realization. 
As analyzed in \cite{chauhan2024}, data-conditioned hypernetworks generate weights based on input characteristics. This enables the target network to dynamically adjust its behavior with better generalization to unseen data.
In parametric dynamical systems, \cite{vlachas2025} demonstrated that hypernetworks can interpolate in the space of models rather than observations. This captures diverse system behaviors across parameter regimes without per-instance retraining.
In our setting, each input realization induces distinct effective dynamics and an associated optimal transformation map and its inverse.
The hypernetwork learns a smooth mapping from input history to neural network parameters, enabling adaptation across $\mf U$. 
During inference, a new input requires only a forward pass through the GRU, with no retraining needed. 
The results confirm that both variants achieve consistent improvements across four qualitatively different input regimes that were not individually optimized for, demonstrating a general input-to-observer mapping.

\section{Conclusion}

We presented a hypernetwork-based framework (HyperKKL) for designing KKL observers for controlled nonlinear systems. We proposed two methods: $\text{HyperKKL}_\text{obs}$, which learns static transformation mappings via encoder-decoder networks and employs input injection in the observer dynamics, and $\text{HyperKKL}_\text{dyn}$, which learns dynamic transformation mappings via weight conditioning of encoder and decoder by a hypernetwork. 
We derived worst-case error bounds that explicitly relate the encoder-decoder approximation errors, the PDE residual, and the decoder's Lipschitz constant.
Through numerical experiments, we demonstrated that input conditioning is necessary for learning the transformation map and its inverse for a KKL observer. 
Future work will investigate spectral normalization to bound the decoder's Lipschitz constant, hybrid architectures that dynamically condition only the decoder, and extensions to unmeasured disturbances via a disturbance estimation module.


\bibliographystyle{IEEEtran}
\bibliography{biblio}
\end{document}